\documentclass[preprint,11pt,authoryear]{elsarticle}
\usepackage[margin=1in]{geometry}
\usepackage[T1]{fontenc}
\usepackage{lmodern,microtype,amsmath,amssymb,amsthm,booktabs,graphicx,longtable}
\usepackage{tabularx,array}
\usepackage{float}
\usepackage[section]{placeins}
\usepackage{etoolbox}
\pretocmd{\subsection}{\FloatBarrier}{}{\PackageError{float_layout}{Cannot install subsection barrier}{}}
\let\CSDAfigure\figure
\renewcommand{\figure}[1][htbp]{\CSDAfigure[H]}
\let\CSDAtable\table
\renewcommand{\table}[1][htbp]{\CSDAtable[H]}

\makeatletter\let\ps@pprintTitle\ps@plain\makeatother
\biboptions{round,authoryear}
\usepackage[colorlinks=true,allcolors=blue,pdfauthor={Aurelien Nicosia}]{hyperref}
\newtheorem{theorem}{Theorem}
\newtheorem{proposition}{Proposition}
\newcommand{\E}{\mathbb E}
\newcommand{\R}{\mathbb R}
\newcommand{\VM}{\operatorname{VM}}
\newcommand{\argu}{\operatorname{Arg}}

\graphicspath{{figures/}}

\begin{document}
\begin{frontmatter}
\title{Identifiability and likelihood estimation for homogeneous angular regression with random coefficients}
\author[ulaval]{Aur\'elien Nicosia\texorpdfstring{\corref{cor}}{}}
\cortext[cor]{Corresponding author.}
\ead{nicosia.aurelien@gmail.com}
\affiliation[ulaval]{organization={Department of Mathematics and Statistics, Universit\'e Laval},
 addressline={1045, avenue de la M\'edecine},city={Qu\'ebec},state={Qu\'ebec},
 postcode={G1V 0A6},country={Canada}}
\begin{abstract}
Repeated directional responses may vary between individuals in both overall
orientation and sensitivity to explanatory cues. We study homogeneous angular
regression with correlated Gaussian cue coefficients and a circular random
intercept, separating observation concentration from resultant length.
Constructive identification conditions recover the joint coefficient law,
intercept concentration and observation-error parameters, including unknown
uniform contamination. The argument combines repeated contrasts, an arc-support
restriction and joint Fourier deconvolution, and extends to exogenous continuous
designs. Integration in coefficient coordinates yields smooth population scores
even when resultant cancellation prevents a regular latent mode. A fixed-proposal
adaptive importance algorithm evaluates the marginal likelihood and its
derivatives. In 8,000 independent simulated datasets, 7,998 fits pass numerical
checks, and 12,799 of 12,800 local profile intervals are available. Correlation
profile-set coverage nevertheless ranges from 91.9\% to 94.8\% at the nominal
95\% level. Greater precision and numerical availability therefore do not ensure
calibrated inference. An exploratory extension to 150 Canadian weather stations is evaluated on
259,807 three-hour forecasts from a new 2026 period after fitting on
2023--2025 data. With circular station rotations in all matched models,
independent station slopes give modest predictive gains over slopes shared
within regions, in all five regions. Estimating their correlations adds no
useful predictive gain; this does not assess parameter uncertainty, and
predictive calibration remains imperfect. The application uses a noncancelling
mean and Laplace integration, so it illustrates the model class rather than
testing the proposed algorithm.
\end{abstract}
\end{frontmatter}

\section{Introduction}
Repeated directional responses often reflect several competing cues. Individuals
may differ both in their overall orientation and in the strength of their
response to each cue. A circular random intercept describes a common rotation;
random cue coefficients describe changes in the shape of the directional
response. Distinguishing these effects from observation noise requires an
identifiable model and a workable marginal likelihood.

We study the homogeneous angular regression of \citet{RivestJRSSC2016}, in which
a normalized planar resultant determines the mean direction and a separate
parameter determines observation concentration. \citet{RivestKatoCJS2019}
introduced a von Mises random intercept, integrated it analytically, and
proposed cluster-specific regression coefficients as an extension in their
Section~7 (p.~726). Here those
coefficients have a joint Gaussian distribution with unrestricted covariance.
The observation family includes both von Mises errors and their uniform
contamination extension.

Related circular mixed models use projected-normal responses
\citep{NunezGutierrez2014,CremersKlugkist2018,CremersEtAl2021} or semiparametric
random coefficients \citep{Maruotti2016}. Bayesian predictive model selection
for circular random effects is studied by \citet{CamliKalaylioglu2021}.
Individual cue preferences also appear in noisy
vector-sum models \citep{ShaverdianEtAl2022}, and random concentration effects
in \citet{NguyenYurko2025}. Recent Bayesian developments include link-adjusted
circular regression with integrated nested Laplace approximation
\citep{YeEtAl2026INLA} and penalizing-complexity priors for circular parameters
\citep{YeEtAl2026PC}. General-purpose software also fits Bayesian von Mises
regression with correlated group-level effects through a tangent-half link
\citep{Burkner2017}. These constructions differ in their response laws and
latent effects. Our focus is the combination of a homogeneous angular response,
a correlated coefficient law and a separate circular rotation.

Two questions drive the development. First, can the observable response law
separate coefficient heterogeneity, rotation and observation error?
Theorem~\ref{thm:identification} gives a constructive answer under explicit
repeated-contrast or exogenous design-support conditions. Repeated contrasts
identify unknown contamination and concentration; an arc-support argument then
separates the rotation and recovers the joint coefficient law. The argument uses
circular deconvolution \citep{DiMarzioEtAl2022}, while differing from
unknown-error deconvolution based on an auxiliary error sample
\citep{JohannesSchwarz2013}, dependent multivariate signals
\citep{GassiatEtAl2022}, or Euclidean repeated measurements
\citep{CapitaoEtAl2024}. These are comparisons of identification mechanisms;
our proofs do not invoke their convergence-rate results.

Second, how can likelihood estimation handle cancellation of the resultant?
Near cancellation, the angular mean changes sharply and a regular latent mode
need not exist. Integrating in coefficient coordinates nevertheless gives
smooth population mean and covariance scores. We combine the complete-data
score identity \citep{Louis1982} with fixed-proposal importance integration
\citep{GeyerThompson1992,Elvira2019}. The statistical issue is the justification
of these established tools for the angular geometry, including the distinction
between integrated smoothness and latent-mode regularity.

Sections~\ref{sec:model}--\ref{sec:estimation} develop the model, identification
and estimation. Section~\ref{sec:numerical} compares integration rules and
Section~\ref{sec:simulation} evaluates estimation and interval coverage.
Section~\ref{sec:wind} examines station heterogeneity in Canadian wind forecasts,
using a noncancelling extension with separately modelled dispersion. The
simulation and application results distinguish numerical availability from
inferential and predictive calibration.

\section{Model and interpretation}\label{sec:model}
For $m$ independent individuals indexed by $i=1,\ldots,m$, let
$Y_{ij}$ be the direction of record $j=1,\ldots,n_i$, where $n_i$ is the
number of records for individual $i$. For $r=0,\ldots,q$, let $x_{ijr}$ and
$z_{ijr}$ denote an observed target angle and multiplier. Index $r=0$
denotes the reference contribution; $q\geq1$ is the number of variable cue
coefficients.
With $v(x)=(\cos x,\sin x)^\mathsf T$, define
\begin{equation}\label{eq:consensus}
 \gamma_{ij}(c)=z_{ij0}v(x_{ij0})+
       \sum_{r=1}^{q}c_rz_{ijr}v(x_{ijr}),\qquad
 \mu_{ij}(c)=\argu\{\gamma_{ij,1}(c)+\mathrm i\gamma_{ij,2}(c)\}.
\end{equation}
The reference coefficient is fixed to one, removing an arbitrary common positive
scale. This reference must have a meaningful, consistently scaled contribution.
We assume $\Pr\{\gamma_{ij}(C_i)=0\}=0$ for every record. A nonconstant affine
consensus and a nondegenerate Gaussian coefficient law satisfy this condition;
a constant consensus must be nonzero. Values assigned on null cancellation
sets do not affect the response law.

The individual effects and responses satisfy
\begin{equation}\label{eq:model}
 C_i\sim N_q(\beta,\Sigma),\quad A_i\sim\VM(0,\kappa_a),\quad
 E_{ij}\sim\VM(0,\kappa_e),\qquad
 Y_{ij}=\mu_{ij}(C_i)+A_i+E_{ij}\pmod{2\pi}.
\end{equation}
The effects and errors are mutually independent, except for observations sharing
$C_i,A_i$. The von Mises density is
$\exp(\kappa\cos t)/\{2\pi I_0(\kappa)\}$, with $I_h$ the modified Bessel
function of order $h$. The population parameter is
$\theta=(\beta,\Sigma,\kappa_a,\kappa_e)$, where $\Sigma$ is positive definite
and both concentrations are positive and finite. We condition on a fixed design,
or on an exogenous random design independent of effects and errors.
Section~\ref{sec:wind} states a separate conditioning model for lagged weather
cues.

\subsection{Response family and extensions}
In the homogeneous formulation, $\kappa_e$ is independent of the covariates and
resultant length. A consensus-error model instead makes concentration depend
on $\|\gamma_{ij}(C_i)\|$ \citep{RivestJRSSC2016}. A projected-normal model
with isotropic Cartesian observation error approaches a uniform angular law
as its Cartesian mean tends to zero. Model~\eqref{eq:model} retains its
conditional concentration near cancellation. This distributional assumption
is relevant whenever competing cues can affect precision as well as direction.

We also consider the homogeneous contamination family
\begin{equation}\label{eq:vmum-main}
 f_E(e;\kappa_e,\varepsilon)=
 \frac{(1-\varepsilon)\exp\{\kappa_e\cos e\}/I_0(\kappa_e)
                   +\varepsilon}{2\pi},\qquad 0\leq\varepsilon<1.
\end{equation}
This adds a common parameter $\varepsilon$ to $\theta$ and reduces to
\eqref{eq:model} at zero contamination. The uniform component accommodates
broad deviations without assigning them a physical cause. The weather model
extends the same regression construction but lets dispersion depend on
covariates and an additional station effect (Table~\ref{tab:model-map}).
\begin{table}[tb]
\centering\small
\caption{Models studied and the integration used for their likelihoods.
The weather model extends the regression construction but changes the latent
effects, error parameters and conditioning assumptions.}
\label{tab:model-map}
\begin{tabularx}{\textwidth}{@{}>{\raggedright\arraybackslash}p{.17\textwidth}*{3}{>{\raggedright\arraybackslash}X}@{}}
\toprule
Configuration & Mean and individual effects & Observation error & Integration \\
\midrule
Homogeneous von Mises & Normalized resultant; Gaussian coefficients and von Mises rotation
 & Constant concentration & Analytic rotation integral; numerical coefficient integral \\
Homogeneous contaminated & Same resultant and individual effects
 & Constant concentration and uniform weight & Joint coefficient and rotation integration \\
Canadian weather (Section~\ref{sec:wind}) & Nonzero anchored resultant with additional fixed terms; von Mises station rotation in every model; no (C), three independent (D) or three correlated (F) Gaussian slopes; independent Gaussian precision effect
 & Covariate-dependent concentration and speed-band uniform weights
 & Two- (C) or five-dimensional (D, F) stationwise Laplace; adaptive Gauss--Hermite check at the fitted values (Supplement Section~S6) \\
\bottomrule
\end{tabularx}
\end{table}

\subsection{Population and individual interpretation}
Write $\sigma_r=\sqrt{\Sigma_{rr}}$ and
$\rho_{rs}=\Sigma_{rs}/(\sigma_r\sigma_s)$. The mean coefficient
$\beta_r$ and standard deviation $\sigma_r$
describe between-individual variation in the units of $z_{ijr}$. Negative
Gaussian coefficients reverse a cue contribution. Correlations describe which
coefficient combinations occur together across individuals; they are not
correlations between observed directions. A circular intercept rotates every
response direction of an individual, whereas cue coefficients change the
response curve's shape (Figure~\ref{fig:effect-interpretation}).
\begin{figure}[!tbp]
 \centering\includegraphics[width=\textwidth]{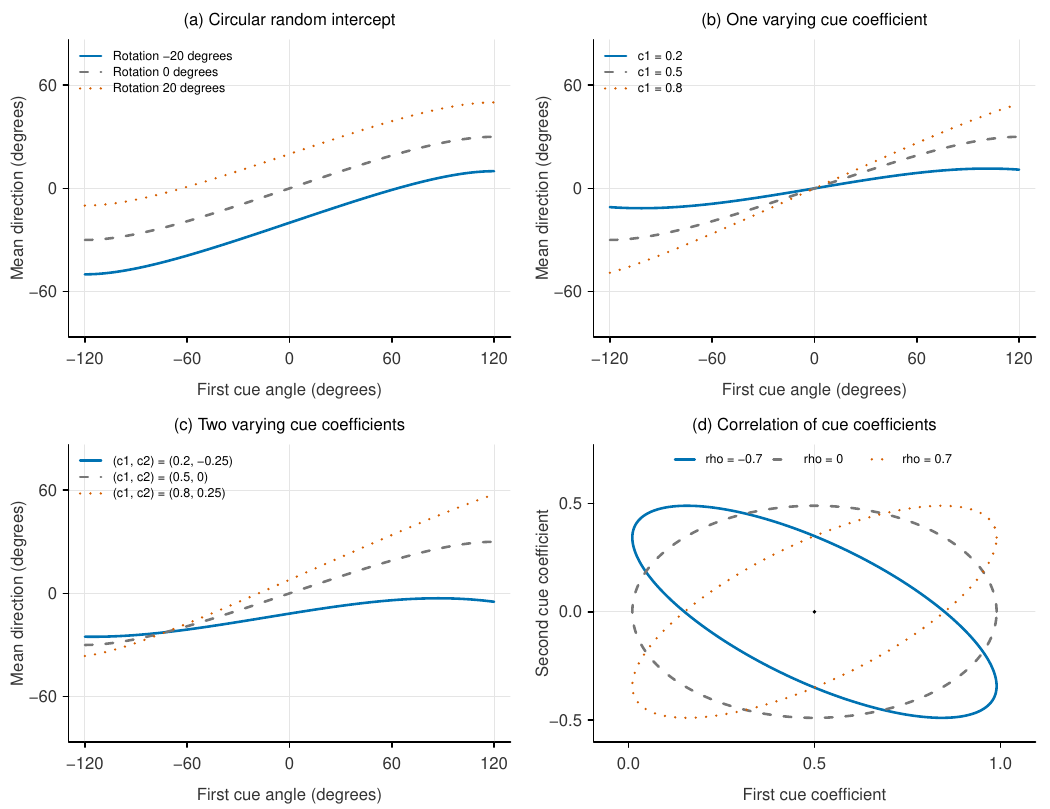}
 \caption{Different roles of individual effects. Panels (a)--(c) use a unit
 reference at zero, a first cue at the horizontal-axis angle, and a second cue
 at 90 degrees, with unit multipliers and no observation noise. The mean is
 the direction of $(1+c_1\cos x,\,c_1\sin x+c_2)$ plus the intercept.
 Panel (a) fixes $(c_1,c_2)=(0.5,0)$; (b) fixes $c_2=0$ and the intercept at
 zero; (c) varies both coefficients with zero intercept. The three profiles
 are illustrative choices, not fitted individuals or random draws.
 Panel (d) shows 95\% Gaussian probability contours with mean $(0.5,0)$,
 marginal standard deviations 0.2 and different correlations. These are
 population-distribution contours, not confidence regions.}
 \label{fig:effect-interpretation}
\end{figure}

The Canadian wind application revisits these mechanisms using fitted station
effects under a specified common covariate scenario
(Figure~\ref{fig:wind-expanded-mechanisms}). Its mean uses the construction
\eqref{eq:consensus} with the previous direction as unit reference, every cue
perpendicular to that reference and additional fixed terms, so the resultant
cannot cancel. Its error law departs from \eqref{eq:model}: concentration and
contamination vary with covariates and a station precision effect. The
application is therefore a related extension, not an instance of the model
covered by Theorem~\ref{thm:identification} or the primary simulations.

These are conditional mixed-model coefficients, not coefficients of a marginal
mean regression. Writing $A_h(k)=I_{|h|}(k)/I_0(k)$, the first moment is
\begin{equation}\label{eq:marginalmoment}
 \E(e^{\mathrm iY_{ij}}\mid x,z)=
 A_1(\kappa_a)(1-\varepsilon)A_1(\kappa_e)
 \E\left\{\frac{\gamma_{ij,1}(C_i)+\mathrm i\gamma_{ij,2}(C_i)}
                   {\|\gamma_{ij}(C_i)\|}\right\}.
\end{equation}
Set $\varepsilon=0$ for pure von Mises errors. Averaging the normalized vector
can change both its direction and length, so the marginal direction need not
be $\mu_{ij}(\beta)$. Rotation and observation noise attenuate the first-moment
length. Consequently, fitted effects are best interpreted through specified
angular contrasts or response curves alongside coefficient estimates.

\section{Constructive population identifiability}\label{sec:identification}
Replication and design variation are needed to separate the effects. If all
targets coincide, the reference multiplier is one and the same non-reference
multiplier vector $z$ is used throughout, the consensus depends on $C_i$
only through $1+z^\mathsf TC_i$.
For $q>1$ in the Gaussian model, distinct $(\beta,\Sigma)$ with the same
$z^\mathsf T\beta$ and $z^\mathsf T\Sigma z$ therefore give the same responses.
With one record per group, interchanging the von Mises intercept and error
concentrations also leaves their convolution unchanged.
The following sufficient conditions avoid these ambiguities and identify a
coefficient law $G$ without assuming Gaussianity.

\begin{theorem}[Constructive identification]\label{thm:identification}
Let $q\geq1$, $C\sim G$ on $\mathbb R^q$, and
$A\sim\VM(0,\kappa_a)$ independently of $C$. Let the observation errors be independent with
density \eqref{eq:vmum-main}, independently of $C,A$, where
$0<\kappa_a,\kappa_e<\infty$ and $0\leq\varepsilon<1$.
The intercept location is fixed at zero. Let $a_1,\ldots,a_q$ be linearly
independent and let $B$ have rows $a_j^{\mathsf T}$.
\begin{enumerate}\renewcommand{\labelenumi}{(\roman{enumi})}
\item For a fixed design, suppose a group contains $q+1$ records with consensus
vectors $(1,a_1^{\mathsf T}C)^{\mathsf T}$ twice and
$(1,a_j^{\mathsf T}C)^{\mathsf T}$ once for $j=2,\ldots,q$.
The joint response law identifies $(G,\kappa_a,\kappa_e,\varepsilon)$.
\item For a parameter-independent random design $D$, suppose $D$ is independent
of $C,A$ and the errors, and groups have at least $\max(2,q)$ records.
Suppose the affine consensus vectors are continuous in the design and, for
$P_D$-almost every design, have probability zero of vanishing under both
parameter values being compared. The support
of the two-record design marginal contains the repeated $a_1$ configuration,
and that of the $q$-record marginal contains the contrast block with matrix
$B$. Then the joint law of $(D,Y)$ identifies
$(G,\kappa_a,\kappa_e,\varepsilon)$. The two support configurations need not
coexist in a realized group or have positive point probability.
\end{enumerate}
The specified identifying configurations have nonzero consensus for every $C$.
Known record-specific rotations may be removed from the response and target
angles. The unit-multiplier construction below has reference angle zero and
other target angles $\pm\pi/2$ in these coordinates.
In the Gaussian submodel, the identified $G$ determines $\beta$ and the full
covariance $\Sigma$. These are sufficient conditions, without a minimality claim.
\end{theorem}

The proof, given in Supplement Section~S1, has three steps. Put
$b_0=1$ and $b_h=(1-\varepsilon)A_h(\kappa_e)$ for $h\ne0$.
For the repeated pair, shared coefficient and rotation effects cancel, so
\begin{equation}\label{eq:identify-error}
 D_h=\E\{\cos[h(Y_1-Y'_1)]\}=b_h^2,\qquad
 \frac{b_2}{b_1}=\frac{I_2(\kappa_e)}{I_1(\kappa_e)},\qquad
 \varepsilon=1-\frac{b_1}{A_1(\kappa_e)}.
\end{equation}
The strict increase of $I_2/I_1$ identifies $\kappa_e$ and then
$\varepsilon$, including its zero boundary.

Next let $U_1=\arctan(a_1^\mathsf TC)$ and remove the identified error
convolution. The Fourier coefficients of $V=U_1+A$ are
$v_h=\E(e^{\mathrm ihY_1})/b_h$. Because $U_1$ lies in an open semicircle,
its law vanishes on a nonempty open arc. The Fourier-tail lemma gives
\begin{equation}\label{eq:identify-rotation}
 \kappa_a=2\limsup_{h\to\infty}\{h!\,|v_h|\}^{1/h}.
\end{equation}
Finally, retain one response for each contrast. With
$U_j=\arctan(a_j^\mathsf TC)$ and $h\in\mathbb Z^q$,
\begin{equation}\label{eq:identify-joint}
 \E(e^{\mathrm i h^\mathsf TY})=
 A_{\sum_jh_j}(\kappa_a)\prod_j b_{h_j}\,
 \E(e^{\mathrm i h^\mathsf TU}),\qquad
 C=B^{-1}(\tan U_1,\ldots,\tan U_q)^\mathsf T.
\end{equation}
All dividing factors are positive. Joint Fourier inversion and the one-to-one
tangent map recover $G$, including coefficient dependence. Conditional Fourier
moments are continuous at the identifying configurations, which transfers the
argument to the design-support conditions in part (ii).

Unit multipliers suffice: take $a_1=(1,\ldots,1)^\mathsf T$ and
$a_j=a_1-2e_j$ for $j\geq2$. Then $\det B=(-2)^{q-1}$; targets at
$\pm\pi/2$ relative to a unit reference at zero realize these contrasts.
For $q=2$ the fixed design uses rows $(1,1),(1,1),(1,-1)$.

Equations~\eqref{eq:identify-error}--\eqref{eq:identify-joint} establish
population uniqueness, not a stable sample inversion or likelihood asymptotics.
At $\varepsilon=1$ or $\kappa_e=0$, uniform errors erase information.
At $\kappa_a=0$ the proof fails, without establishing nonidentifiability for
all designs. Exogeneity is substantive: lagged movement or weather cues do not
satisfy it merely because their realized values vary.

\section{Marginal likelihood and estimation}
\label{sec:estimation}\label{sec:likelihood}
\subsection{Integrated likelihood and Gaussian scores}
Let $n_i$ be the number of observations in group $i$, and define
\[
 S_{i1}(c)=\sum_j\cos\{Y_{ij}-\mu_{ij}(c)\},\qquad
 S_{i2}(c)=\sum_j\sin\{Y_{ij}-\mu_{ij}(c)\}.
\]
Set $X_i=\kappa_a+\kappa_e S_{i1}$, $Z_i=\kappa_e S_{i2}$ and
$R_i=(X_i^2+Z_i^2)^{1/2}$. Analytic integration of $A_i$, as in
\citet{RivestKatoCJS2019}, gives the conditional group density
\begin{equation}\label{eq:kernel}
 h_i(c;\kappa_e,\kappa_a)=
 \frac{I_0(R_i(c))}{(2\pi)^{n_i}I_0(\kappa_e)^{n_i}I_0(\kappa_a)}.
\end{equation}
With $\phi_q(c;\beta,\Sigma)$ the Gaussian coefficient density, the remaining
likelihood is
\begin{equation}\label{eq:integral}
 \ell(\theta)=\sum_{i=1}^m\log L_i(\theta),\qquad
 L_i(\theta)=\int_{\R^q}h_i(c;\kappa_e,\kappa_a)
                         \phi_q(c;\beta,\Sigma)\,dc.
\end{equation}
It involves $q$-dimensional Gaussian integration rather than $(q+1)$-dimensional
integration over both types of individual effects. All normalizing constants are retained.

\begin{proposition}\label{prop:smooth}
For fixed finite group sizes and design, suppose the zero-consensus sets do not
have positive Lebesgue measure. On compact subsets with positive-definite
$\Sigma$ and positive finite concentrations, $L_i$ is continuously
differentiable of every finite order in $\beta$ and the entries of $\Sigma$.
In particular, writing $\delta=C_i-\beta$ and conditioning on the observed group,
\begin{align}
 \nabla_\beta\log L_i&=\Sigma^{-1}\E(\delta\mid Y_i),\label{eq:scorebeta}\\
 \nabla_\Sigma\log L_i&=\tfrac12\Sigma^{-1}
          \{\E(\delta\delta^\mathsf T\mid Y_i)-\Sigma\}\Sigma^{-1}.
          \label{eq:scoresigma}
\end{align}
The matrix derivative in \eqref{eq:scoresigma} is defined by the trace inner
product against symmetric perturbations of $\Sigma$.
\end{proposition}
\begin{proof}
The conditional density in \eqref{eq:kernel} is nonnegative, measurable and
bounded uniformly in $c$ when concentrations range over a compact set. It is
independent of $\beta$ and $\Sigma$ in coefficient coordinates. Derivatives of a
Gaussian density of any fixed finite order are a Gaussian density times a
polynomial in $c$, with coefficients bounded on the stated parameter sets.
They admit an integrable envelope of the form
$K(1+\|c\|^r)\exp(-a\|c\|^2+b\|c\|)$ for suitable $a>0$.
Differentiation under the integral and continuity follow by domination.
Dividing the first derivatives by the positive integral gives the two Gaussian
score identities.
\end{proof}

The argument also covers the mixture-error kernel: it is bounded, measurable
and independent of $(\beta,\Sigma)$ at fixed coefficient coordinates.
Thus cancellation can obstruct a latent-mode approximation without obstructing
these integrated derivatives. For $\Sigma=LL^\mathsf T$, chain rules give the
Cholesky scores; weighted first and second coefficient moments suffice.
Numerical nodes with exactly zero consensus are rejected.

This smoothing is specific to parameters in the Gaussian density. If a fixed
Cartesian contribution contains a component outside the random-coefficient
design span, its derivative still acts on the angular kernel. Even a finite
first derivative can then fail at rank-one cancellation. Supplement Section~S2
proves this counterexample and gives the exact coefficient translation for
spanned fixed terms.

For the concentrations, define $B(R)=A_1(R)/R$ with $B(0)=1/2$. The conditional
scores on logarithmic scales are
\begin{align}
 s_e(c)&=\kappa_e\{B(R_i)(X_iS_{i1}+Z_iS_{i2})-n_iA_1(\kappa_e)\},\\
 s_a(c)&=\kappa_a\{B(R_i)X_i-A_1(\kappa_a)\}.
\end{align}
Their conditional expectations give the marginal concentration scores, with
$B(0)=1/2$ removing the apparent singularity at zero resultant length.

For contaminated errors, the remaining integral includes the rotation:
\begin{equation}\label{eq:joint-likelihood}
 L_i(\theta)=\int_{\mathbb R^q}\int_{-\pi}^{\pi}
 \prod_j f_E\{y_{ij}-\mu_{ij}(c)-a;\kappa_e,\varepsilon\}
 \frac{e^{\kappa_a\cos a}}{2\pi I_0(\kappa_a)}
 \phi_q(c;\beta,\Sigma)\,da\,dc.
\end{equation}
The Gaussian scores remain unchanged. With complete-data score $s_i$ and
log-likelihood Hessian $H_i$, the observed information is
\begin{equation}\label{eq:information}
 \mathcal I_i=-\E(H_i\mid Y_i)-\operatorname{Cov}(s_i\mid Y_i),
\end{equation}
as in \citet[p.~227, equations (3.1)--(3.2)]{Louis1982}. Supplement Section~S3 gives the conditional error
derivatives and the analytically integrated one-sided score at
$\varepsilon=0$. This boundary is included in estimation; ordinary interior
Wald intervals are withheld there.

\subsection{Fixed-proposal adaptive integration}
\label{sec:integration-algorithm}
Write $t_i(u;\theta)$ for the group integrand: $u=c$ in
\eqref{eq:integral}, or $u=(c,a)$ in \eqref{eq:joint-likelihood}.
At a proposal refresh, construct densities $g_{ia}$ with weights $\alpha_a$.
Transform a randomized Sobol rule to obtain $K$ nodes $u_{iak}$ from each
component. The likelihood approximation is
\begin{equation}\label{eq:mis}
 \widehat L_i(\theta)=\sum_a\sum_{k=1}^K w_{iak}(\theta),\qquad
 w_{iak}(\theta)=\frac{\alpha_a}{K}
 \frac{t_i(u_{iak};\theta)}{\sum_b\alpha_b g_{ib}(u_{iak})}.
\end{equation}
Equal component node counts retain the unequal mixture weights. The likelihood
is not divided by an estimated weight sum. Normalized weights
$\widetilde w_{iak}=w_{iak}/\widehat L_i$ are used only for moments and scores:
\begin{equation}\label{eq:numerical-score}
 \nabla\log\widehat L_i=\sum_{a,k}\widetilde w_{iak}
                  \nabla\log t_i(u_{iak};\theta).
\end{equation}
Nodes and proposal densities remain fixed throughout each local optimization,
so this is the derivative of that numerical objective. Proposals are refreshed
between optimizations. This differs from differentiating moving adaptive
quadrature nodes \citep{Stringer2025}, and from an EM update.

For the joint mixture-error calculation, the proposal is a mixture of a
wrapped joint Gaussian approximation to the conditional effects and the
population Gaussian coefficient density times a uniform angle, with weights
0.8 and 0.2. Its local covariance is inflated by 1.5 and augmented by 0.05
times the population scale matrix. These terms modify the integration proposal,
not the statistical likelihood. Algorithm~1 summarizes the primary study's
implemented estimator; Supplement Section~S3 supplies initialization formulas,
search limits and refinement controls. Pure von Mises calculations use the
same principle after analytic rotation integration, with Gaussian coefficient
proposals and their separately specified controls.

\begin{table}[!tbp]
\centering
\begin{minipage}{.97\textwidth}\small
\hrule\medskip
Algorithm 1. Adaptive marginal-likelihood estimation for the contaminated model
\medskip\hrule\medskip
Input: grouped responses and design, two declared starting vectors, parameter
bounds, integration seeds and numerical tolerances. Parameterize
$\Sigma=LL^\mathsf T$ with log Cholesky diagonals, log concentrations and direct
$\varepsilon\in[0,1)$.
\begin{enumerate}\setlength{\itemsep}{3pt}
\item For each start, construct the local wrapped-Gaussian and population/uniform
proposals. At powers $p=8,10,12$, use $K=2^p$ nodes per component.
\item With nodes fixed, minimize $-\sum_i\log\widehat L_i$ by bounded
quasi-Newton optimization using \eqref{eq:numerical-score}. Update conditional
moments and refresh proposals, with at most 7, 5 and 4 local optimizations at
the respective powers; stop a stage early when the maximum coordinate change
is below $10^{-5}$.
\item Evaluate each candidate with two independent randomized shifts, initially
at power 13. Check free-coordinate information, standardized score corrections,
likelihood agreement, effective sample sizes and active limits. Apply the
one-sided score check if $\widehat\varepsilon=0$.
\item Refine only according to the declared numerical policy. If neither
initial fit passes, continue failed candidates in increasing negative-log-likelihood
order within their additional budgets, up to validation power 16, stopping at
the first accepted continuation.
\item Select the accepted candidate with the smallest validated negative log
likelihood. If no candidate passes, return numerical failure. Retain all
attempts and their status.
\end{enumerate}
Output: parameter estimates, conditional moments, information and numerical
status. Inference is evaluated separately.
\medskip\hrule
\end{minipage}
\end{table}

\subsection{Numerical acceptance and interval estimation}
Let $H$ be the negative log-likelihood Hessian on free coordinates and $g_r$
an independently evaluated negative log-likelihood gradient. The diagnostic
\begin{equation}\label{eq:correction}
 d_r=\max_j\frac{|(H^{-1}g_r)_j|}{\{(H^{-1})_{jj}\}^{1/2}}
\end{equation}
approximates a Newton correction in standard-error units. Primary-study
acceptance requires positive usable information, every $d_r<0.05$, a validation
log-likelihood range below 0.05, minimum group effective sample size
$(\sum w)^2/\sum w^2\geq100$, and no prohibited active limit or timeout.
At exact zero contamination, both independent shifts must also satisfy the
one-sided score tolerance. These diagnostics assess numerical stability;
they are not error bounds or certificates of a global maximum. Numerical
error rates sufficient for likelihood asymptotics require additional conditions
\citep{BilodeauStringerTang2026}.

For a natural scalar parameter $\psi$, define
\begin{equation}\label{eq:profile-set}
 W(\psi)=2\{\ell(\widehat\theta)-\ell(\widehat\theta_\psi)\},\qquad
 \mathcal C_c=\{\psi:W(\psi)\leq c\},
\end{equation}
where $\widehat\theta_\psi$ denotes a nuisance-parameter maximizer at fixed
$\psi$. The implemented constrained optimizer uses safeguarded Newton steps,
curvature from \eqref{eq:information}, a quasi-Newton fallback and independent
integration checks. Root searches around the selected maximum return local
candidate intervals, which need not exhaust $\mathcal C_c$.

We assess the fixed cutoff $c=3.84145882069412$, the 95th percentile of
$\chi^2_1$, rather than assuming its finite-sample validity.
Interior Wald intervals use transformed observed information: log scales for
standard deviations and concentrations, the Fisher scale for correlations,
and the logit scale for contamination. They are withheld
at estimated boundaries. Interior and boundary likelihood limits require
regularity beyond population identification
\citep{VanderVaart1998,SelfLiang1987}; the next study evaluates the resulting
procedures directly.

Numerical acceptance is a status of the computed estimate, not a confidence-level
guarantee. The profile sets and local intervals below are evaluated procedures;
the evidence does not support routinely interpreting correlation intervals as
calibrated 95\% intervals, particularly with few groups. We retain their fixed
cutoff and report the calibration results rather than applying an empirical
post-evaluation correction. The implementation documentation separates this
inferential status from the numerical checks.

\section{Numerical integration assessment}\label{sec:numerical}
Figure~\ref{fig:geometry} shows conditional coefficient densities near
resultant cancellation. The two groups were selected for their nearest
cancellation locations in the fitted Gaussian metric, illustrating a numerical
difficulty rather than its prevalence.
\begin{figure}[tb]
 \centering\includegraphics[width=\textwidth]{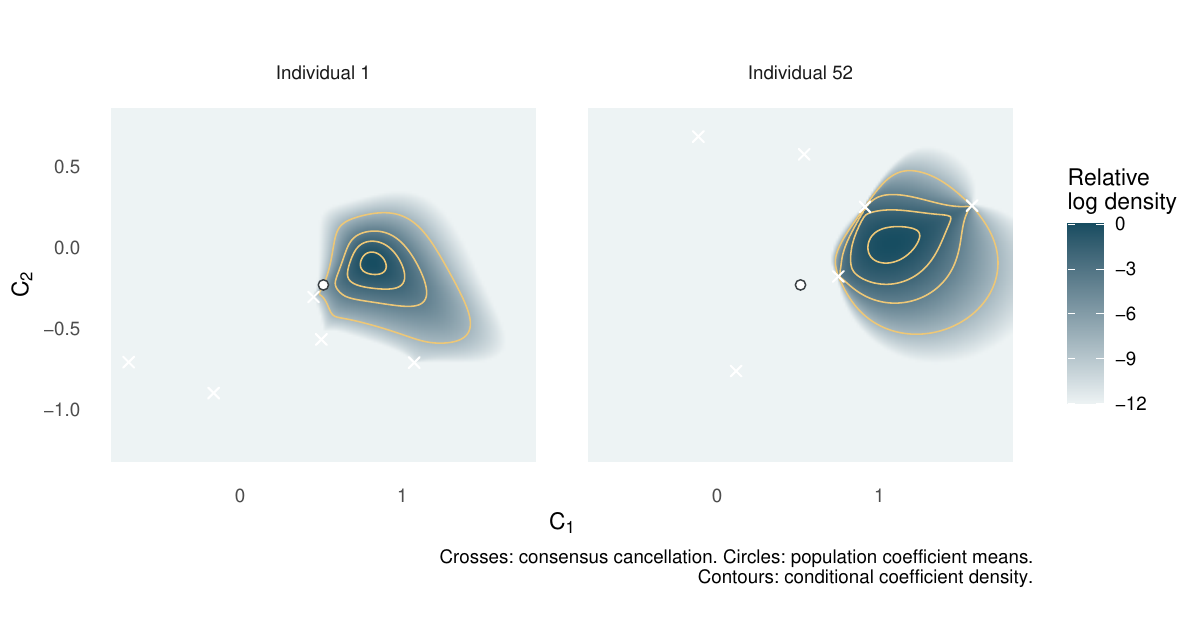}
 \caption{Conditional coefficient densities in a two-slope simulated dataset.
 Crosses mark coefficients at which an observed consensus would cancel; circles
 mark the population coefficient mean. The two groups have the nearest such
 locations to the population mean in the fitted Gaussian metric.}
 \label{fig:geometry}
\end{figure}

Independent checks compare analytic scores with finite differences, node
densities with separate implementations, and a scalar integral with adaptive
integration. A separately compiled Template Model Builder (TMB) joint-density template
\citep{Kristensen2016TMB} supplies a Laplace comparator after gradient and
Hessian checks. With the same smooth mode and curvature, Laplace is one-point
adaptive Gauss--Hermite quadrature \citep{LiuPierce1994,PinheiroBates1995}.
The experiment below also assesses higher-order rules. Full fitting benchmarks
and additional-fixed-term comparisons are summarized in Supplement Section~S8
and detailed in its indexed scientific archive.

\subsection{Comparison of integration rules}
A fixed-parameter experiment compares the same pure von Mises marginal
likelihood and population-score integrals, with rotation integrated analytically.
It uses replicates 1--3 of six earlier designs: one-, two- and three-coefficient
reference cases and three-coefficient large-variance, high-correlation and
nearly aligned cases. Each has 60 groups of 15 responses. Generating parameters,
case identities and numerical rules are specified in Supplement Section~S3.
These are distinct from the contaminated-model study in Section~\ref{sec:simulation};
parameters are fixed at their generating values rather than optimized.

Seventeen of the 18 references pass their declared stability checks. The
unresolved aligned-design case is retained but excluded from accuracy summaries.
The single-mode AGHQ construction accepts a smooth mode for every group in
six cases, five of which have a stable reference. Failed mode searches do not
prove that no regular mode exists. Table~\ref{tab:integration_comparison} reports
both all available cases and the same five cases for every rule.
\begin{table}[tb]\centering\small
\caption{Absolute differences in total negative log likelihood from the saved numerical reference for the pure von Mises model. All available cases and the common five cases are shown separately. Medians pool case evaluations: one per deterministic rule and three independent shifts or seeds per stochastic rule. Maxima retain individual evaluations. Resolution is per axis for GH and AGHQ, and $2^{14}$ per proposal component for stochastic rules. Reference variation and costs are in Supplement Section~S3.}
\label{tab:integration_comparison}
\setlength{\tabcolsep}{4pt}\begin{tabular}{lrrrrr}\toprule
 & & \multicolumn{2}{c}{All available cases} & \multicolumn{2}{c}{Common five cases} \\
Rule & Cases & Median & Maximum & Median & Maximum \\ \midrule
Prior GH, 41/axis & 17/17 & 0.5188 & 3.8073 & 1.0569 & 1.4734 \\
AGHQ, 1/axis & 5/17 & 2.2688 & 78.0023 & 2.2688 & 78.0023 \\
AGHQ, 25/axis & 5/17 & 0.0537 & 26.1324 & 0.0537 & 26.1324 \\
Prior Monte Carlo & 17/17 & 0.3329 & 2.9585 & 0.4064 & 2.9585 \\
Prior Sobol & 17/17 & 0.0483 & 1.1389 & 0.0630 & 1.1389 \\
Mixture Sobol & 17/17 & 0.0017 & 0.0355 & 0.0005 & 0.0355 \\
\bottomrule\end{tabular}\end{table}

The reference is the mean of three higher-resolution randomized evaluations;
its observed negative-log-likelihood ranges over the resolution checks are
$3.55\times10^{-6}$ to 0.01496 in the 17 retained cases. These ranges are not
error bounds. The small discrepancies of adapted integration measure agreement
with this reference, not certified absolute accuracy. All rules share a compiled
kernel; independent randomizations therefore do not exclude shared numerical
bias. The separate scalar and derivative checks address different parts of the
calculation.

At these resolutions, adaptation reduces reference discrepancies on the common
cases but uses more nodes than prior Monte Carlo or Sobol and incurs proposal
preparation. Supplement Section~S3 reports case-level differences, score errors,
node counts and preparation/evaluation costs. This experiment supports neither
universal superiority at equal work nor a comparison of optimized estimators
or joint contaminated-model integrators.

\section{Simulation study}\label{sec:simulation}
\subsection{Design and assessment targets}
The primary study evaluates the two-coefficient model with contaminated errors
in eight scenarios of 1,000 datasets each (Table~\ref{tab:vmum-independent-design}).
Independent uniform target directions and multipliers on $(0.5,1.5)$ are
independent of effects and errors, satisfying Theorem~\ref{thm:identification}(ii).
The reference scenarios increase the number of individuals from 23 to 60 and
100 while keeping 30 responses per individual. Other scenarios change group
length and concentration, contamination, or coefficient spread and correlation.
Contamination is estimated even when its generating value is zero.
Because Long changes both group length and concentrations, it does not isolate
an effect of group length alone.
\begin{table}[tb]\centering\small
\caption{Independent mixture-error design and numerical acceptance. Each scenario has two Gaussian coefficients and 1,000 attempted datasets. The means are $(0.6,-0.25)$ and baseline coefficient standard deviations are $(0.45,0.35)$. SD scale multiplies both standard deviations. Ref denotes a reference design; Long changes group length and both concentrations; Zero, Near and Cont change contamination; Corr changes correlation and coefficient spread.}\label{tab:vmum-independent-design}
\begin{tabular}{lrrrrrrr}
\toprule
Scenario & $m$ & $n_i$ & $\varepsilon$ & $\rho_{12}$ & SD scale & $(\kappa_e,\kappa_a)$ & Accepted \\
\midrule
Ref23 &  23 &  30 & 0.10 & 0.35 & 1.0 & (8,5) & 1000/1000 \\
Ref60 &  60 &  30 & 0.10 & 0.35 & 1.0 & (8,5) & 1000/1000 \\
Ref100 & 100 &  30 & 0.10 & 0.35 & 1.0 & (8,5) & 999/1000 \\
Long &  23 & 114 & 0.10 & 0.35 & 1.0 & (15,65) & 999/1000 \\
Zero &  23 &  30 & 0.00 & 0.35 & 1.0 & (8,5) & 1000/1000 \\
Near &  23 &  30 & 0.02 & 0.35 & 1.0 & (8,5) & 1000/1000 \\
Cont &  23 &  30 & 0.20 & 0.35 & 1.0 & (8,5) & 1000/1000 \\
Corr &  23 &  30 & 0.10 & 0.80 & 1.5 & (8,5) & 1000/1000 \\
\bottomrule
\end{tabular}
\end{table}

The estimator, data-only recovery rule and profile cutoff were fixed after a
separate 400-dataset pilot, with disjoint calibration and evaluation seeds.
Every parameter is assessed by a constrained fit at its true value, giving
64,000 profile-set membership calculations. A prespecified endpoint substudy
uses the first 200 datasets per scenario to construct 12,800 local intervals;
truth is used only afterward to assess coverage. Failed fits and unavailable
inference remain in the all-attempt denominators.

We distinguish numerical availability, bias and root mean squared error (RMSE) among accepted fits,
coverage conditional on available inference, and the fraction of all attempts
returning covering inference. At 95\% coverage, the Monte Carlo standard errors
are approximately 0.0069 and 0.0154 for 1,000 and 200 replicates. Wilson limits
are pointwise \citep{Morris2019}. The supplement reports all parameters and
matched comparisons of empirical and model standard errors.

\subsection{Numerical availability and estimation precision}
Numerical checks accepted 7,998 of 8,000 fits, with one rejection each in Ref100
and Long. Of 64,000 true-value profiles, 26 were unavailable. In Zero, 936 of
1,000 contamination estimates were exactly zero and passed the one-sided checks;
four estimated-boundary fits occurred in Near. These fits contribute to point
estimation and profile inference but not ordinary interior Wald inference.

Increasing the number of individuals improved precision. For Ref23, Ref60 and
Ref100, the RMSEs of $\widehat\beta_1$ were 0.0946, 0.0597 and 0.0466;
those of $\widehat\rho_{12}$ were 0.2141, 0.1215 and 0.0975. Ref23 correlation
bias was $-0.0059$ (Monte Carlo standard error 0.0068). Its root mean square
model standard error was nevertheless 0.1953 against empirical standard
deviation 0.2141 on the same subset. Small bias and successful computation
therefore coexist with underestimated sampling variation.

\subsection{Interval calibration}
Across all 64 parameter--scenario cells, conditional profile-set coverage
ranged from 0.919 to 0.991. Correlation coverage was below the nominal 0.95
in all eight scenarios, ranging from 0.919 to 0.948. It was 0.919 in Ref23
(pointwise Monte Carlo limits $[0.900,0.934]$), 0.948 in Ref60 and 0.936 in
Ref100 (Figure~\ref{fig:vmum-independent-coverage}). Unrounded pointwise upper
limits are below 0.95 in Ref23, Ref100, Long, Near and Cont; the Ref100 limit
is 0.94951. These limits are not simultaneously adjusted; with a Holm
adjustment of one-sided exact binomial tests across the eight scenarios,
Ref23, Near and Cont remain below 0.95 at the 5\% level (Supplement
Section~S5). The Ref60--Ref100 difference is not clearly resolved at these
sample sizes.

\begin{figure}[tb]
 \centering\includegraphics[width=\textwidth]{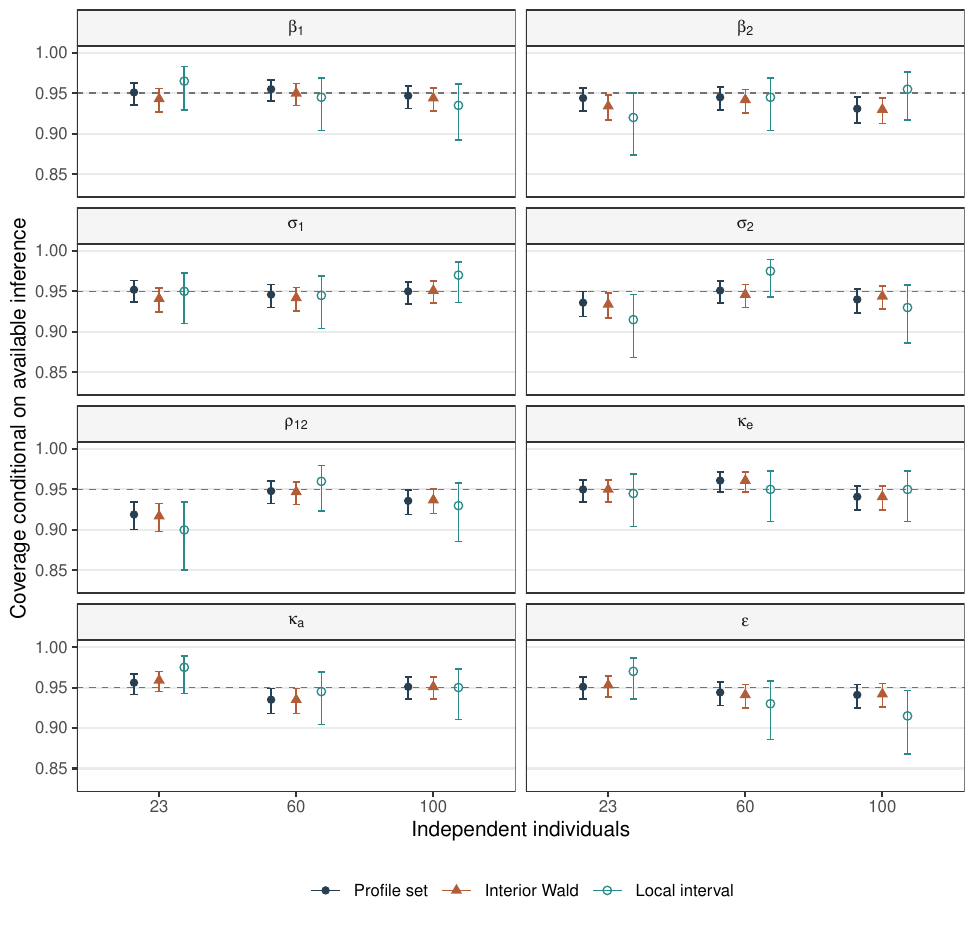}
 \caption{Coverage conditional on available inference in the independent
 reference scenarios. Profile-set and interior Wald assessments use 1,000
 attempted datasets per scenario; local intervals use the prespecified first
 200. Error bars are pointwise 95\% Wilson Monte Carlo limits; the dashed line
 is 0.95. The series are paired calculations. Complete denominators and
 all-attempt success fractions are in Supplement Section~S4.}
 \label{fig:vmum-independent-coverage}
\end{figure}

The endpoint procedure returned 12,799 of 12,800 local intervals. One
$\sigma_2$ interval in Long failed profile-point validation. Correlation
intervals covered in 180/200 Ref23 datasets and 178/200 Near datasets, with
pointwise limits $[0.851,0.934]$ and $[0.839,0.926]$. Median correlation widths
in Ref23, Ref60 and Ref100 were 0.768, 0.475 and 0.375. Among 12,798 cases
with both assessments available, six endpoint and true-value membership
decisions differed; all were within 0.022 of the fixed likelihood-ratio cutoff.
These cases are retained in Supplement Section~S4.

Boundary behavior differs. In Zero, profile-set coverage for $\varepsilon$
was 0.991, indicating conservative behavior for this assessment rather than
nominal calibration. The common profile cutoff is an evaluation choice, not
a shared theoretical justification for interior and boundary regimes.
Its interior Wald intervals were available for only 64 fits and
never included zero. This selected interior subset cannot describe the
unconditional sampling distribution.

A post-evaluation correlation diagnosis (Supplement Section~S5) finds agreement
in all 1,600 paired correlation decisions. In Ref23, the Fisher-scale empirical
standard deviation is 0.2598 versus root mean square model standard error
0.2308. The signed root likelihood ratio has mean 0.061 and standard deviation
1.104; the empirical 95th percentile of the likelihood ratio is 4.821
(pointwise limits $[4.264,5.530]$). Combining extrema of saved independent
likelihood evaluations changes none of the 1,000 Ref23 membership decisions.
The evidence is consistent with finite-sample miscalibration but cannot exclude
shared numerical bias or a competing maximum. No diagnostic quantile replaces
the prespecified cutoff, and no correction is validated.

\subsection{Weak designs and temporal dependence}
Three separate stress experiments use 200 datasets each. Nearly aligned cues
yield only 107 accepted fits; parameter-specific profiles are available in
103--105 cases. Conditional coverage ranges from 0.923 to 0.971, but successful
coverage over all attempts is only 0.480--0.505. The principal failure here is
numerical availability under a weak directional design.

Two further experiments preserve the mixture-error marginal while introducing
Gaussian-copula serial dependence of 0.35. All fits are accepted at both $m=23$
and $m=60$, yet profile coverage of the generating $\kappa_e$ is 0.905 and
0.895. Under this misspecification, a generating parameter need not equal the
fitted model's pseudo-true target. These results are reported separately from
the independent-error evaluation.

Earlier pure von Mises experiments with one to three random coefficients
accepted all 600 reference fits but substantially fewer sparse-group and
nearly aligned cases. Their distinct protocols and the additional-fixed-term
experiments are summarized in Supplement Section~S8, without pooling their
outcomes with this study.

\FloatBarrier
\section{Station-specific responses in Canadian wind directions}\label{sec:wind}
\subsection{Scientific question and temporal evaluation}
How much directional-response heterogeneity remains between stations after
allowing for a circular rotation and regional mean differences? We use
out-of-period forecasts to assess the predictive consequences of fitting
station slopes and their correlations. This criterion does not assess the
uncertainty or coverage of population coefficient and covariance estimates.
We examine these questions at 150 Canadian weather
stations, with 30 stations in each of British Columbia, the Prairies, Ontario,
Quebec and Atlantic Canada. Public hourly observations from \citet{ECCC2023}
record direction clockwise from north in ten-degree units. Our target is the
recorded direction conditional on noncalm wind and availability; all models
assign probabilities to the same 36 ten-degree categories.

The network is selected using 2023--2025 availability and geographical coverage,
without inspecting 2026 data. An initial January 1--14, 2023 prefix supplies
station-history summaries. Final fitting uses 1,166,700 noncalm responses from
January 15, 2023 through December 31, 2025 on a regular three-hour UTC grid.
The fitted parameters, transformations and analysis programs were frozen before
acquiring the evaluation data. The 259,807 common test responses run from
January 1, 2026 at 03:00 UTC through August 31 at 21:00 UTC; every forecast
origin is later than the last estimation response. All 150 selected stations
contribute test observations. This is an exploratory extension motivated by an
earlier Quebec analysis, followed by a new temporal evaluation at known
stations. It is neither a preregistered confirmatory study nor an assessment
of transfer to unseen stations or of a complete annual cycle. Supplement
Section~S6 documents the protocol; Section~S7 retains the earlier study.

Forecasts use measurements recorded no later than their three-hour-earlier origin:
previous directions, speed, temperature, humidity, pressure change and calendar
harmonics. A regional cue combines up to five retained neighbours between 10
and 300 km away, weighted by speed and inverse distance. Direction histories
use the three-hour grid. Missing weather predictors are carried forward from
the past, and reference age enters the model. Each observed post-prefix grid
direction contributes a response factor, including directions subsequently
used as predictors. Station spacing does not imply spatial independence. This retrospective
evaluation does not verify historical publication delays or real-time data
availability.

\subsection{Anchored means and matched station-effect comparisons}
Let $\theta^0_{st}$ be the last available grid direction at the forecast origin
for station $s$, and let
$v_\perp(\theta)=(-\sin\theta,\cos\theta)^\mathsf T$. We use
\begin{equation}\label{eq:wind-mean}
 \gamma_{st}=v(\theta^0_{st})+h_{st}v_\perp(\theta^0_{st}),\quad
 h_{st}=X_{st}^{\mathsf T}\beta+Z_{st}^{\mathsf T}b_s,\quad
 \mu_{st}=\theta^0_{st}+A_s+\arctan h_{st}\pmod{2\pi}.
\end{equation}
The unrotated norm is $\sqrt{1+h_{st}^2}>0$. The three entries of $Z$ are
the perpendicular regional cue, $\sin(2\theta^0)$ and $\cos(2\theta^0)$.
All three matched models include a circular rotation
$A_s\sim\VM(0,\kappa_A)$ and an independent Gaussian station precision effect.
Model C has no station slopes ($b_s=0$): the three slopes are common to the
stations of a region and differ between regions only through fixed contrasts.
Model D has independent Gaussian slopes, $b_s\sim N_3(0,\Sigma)$ with diagonal
$\Sigma$; model F estimates the full $\Sigma$. Rotation, slopes and precision
are mutually independent conditional on the initial history. D versus C
therefore isolates the addition of station slopes with rotation already
present, while F versus D adds only their three population correlations.
Because C already includes random station rotation and precision, a D--C score
gain is not a general comparison of mixed and nonmixed regression. We do not
compare different numbers of random slopes within D or estimate a separate
covariance for each region. D estimates slope standard deviations of 0.151,
0.072 and 0.072 for the regional-cue, sine and cosine terms. For scale, the
regional-cue main effect is 0.350 in the Quebec reference region, with fixed
regional contrasts from $-0.148$ to $0.132$; the fitted rotation concentration
$\kappa_A=221$ corresponds to a rotation standard deviation of about four
degrees. Laplace--Wald standard errors computed from the saved Hessians after
the evaluation are reported in Supplement Section~S6; they are model-based and
have not been evaluated for coverage in this application.

Each model has 46 fixed mean coefficients: the established 30 history,
weather and directional terms, with daily harmonics expressed in local solar
time, plus four regional contrasts interacting with the intercept and the
three slope covariates. Initial-history summaries enter the coefficient means,
following the conditional-modelling strategy of \citet{Wooldridge2005}.
The von Mises/uniform observation family has five mixture weights by
preceding-speed band. Concentration depends on speed, recent directional
change, reference age and region through 14 coefficients and a Gaussian station
effect, separately from resultant length. All parameters are re-estimated for
the expanded dataset. The fitted covariance therefore describes departures
after regional and initial-history mean adjustment.

Under the external-weather and recording-pattern assumptions stated in the
supplement, backward summation over response categories normalizes the
conditional joint law, and the likelihood factors into station integrals even
with lagged neighbour cues. This is a working three-hour conditional model;
normalization does not establish correct specification or the exogeneity
required by Theorem~\ref{thm:identification}. The noncancelling mean allows
stationwise Laplace approximation with automatic derivatives
\citep{Kristensen2016TMB}. Models C, D and F have respectively two, five and
five station-effect dimensions. Numerical controls check gradients, actual
unregularized curvature, alternative-start stability and category quadrature
(Supplement Section~S6). A separate check added after the evaluation compares
every station's Laplace contribution with adaptive Gauss--Hermite quadrature
at the frozen fitted values; summed over stations, the largest absolute
difference is about 0.21 log-likelihood units (Supplement Section~S6). These controls concern
the computed likelihood; they do not establish calibrated parameter inference. Test predictions substitute the frozen
conditional station modes; later observations update predictors, not fitted
parameters or effects.

Figure~\ref{fig:wind-expanded-mechanisms} connects the fitted Canadian model
to the schematic mechanisms in Figure~\ref{fig:effect-interpretation}. It uses
F so that the estimated slope covariance can be displayed, although D is the
more parsimonious model for the prediction comparison below. Recomputing the
three curve panels from the frozen D fit selects the same stations and changes
no displayed curve by more than 0.05 degrees (a check added in this revision).
Under a common standardized Quebec scenario, fitted station rotations shift the
response curves, whereas fitted slopes alter their shapes. The three stations
are selected by fixed quantiles of their training-fit regional-cue modes; the
figure illustrates conditional point estimates, not observed forecast accuracy
or parameter uncertainty. The estimated
regional-cue--sine slope correlation in F is 0.067, leaving its population
contour close to one with the same marginal standard deviations and zero
correlation. This does not establish zero population correlation. The separate
station precision effect changes dispersion and is not drawn.

\begin{figure}[!tbp]
\centering\includegraphics[width=\textwidth]{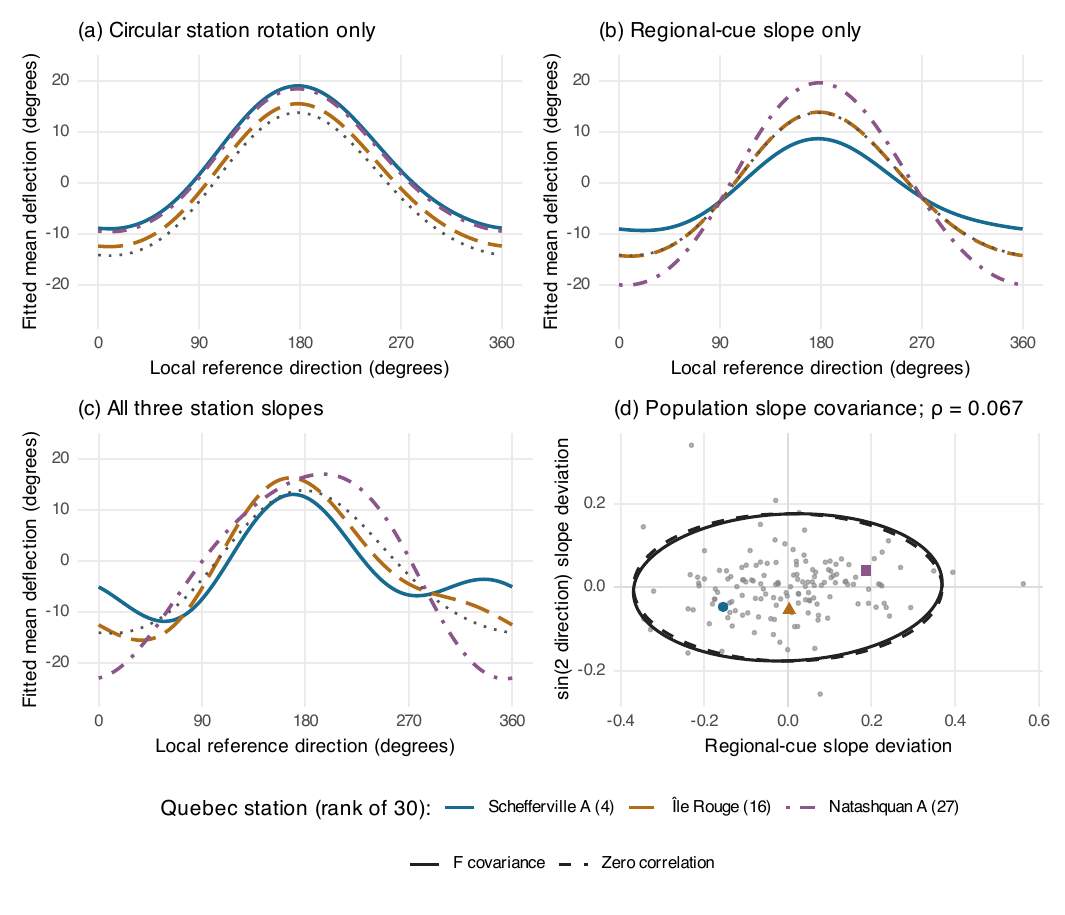}
\caption{Fitted analogue of Figure~\ref{fig:effect-interpretation} for the
150-station Canadian study, using the full-covariance model F fitted on
2023--2025 observations. The blue, ochre and purple curves use conditional effects from Schefferville A, \^Ile Rouge and
Natashquan A, respectively: ranks 4, 16 and 27 of the fitted regional-cue slope
modes among 30 Quebec stations, the nearest ranks to the 10th, 50th and 90th
percentiles. Fixed covariates and
initial-history summaries use the same \^Ile Rouge background; weather terms
are centred, preceding directional change and reference age are zero, and
calendar terms represent 1 June 2025 at 12:00 UTC. A regional cue of length
0.6 points at 270 degrees while the local reference direction varies.
(a) Retain each station's rotation only; (b) retain its regional-cue slope
only; (c) retain all three slopes with rotation zero. Dotted grey curves set
all random mean effects to zero. (d) The solid contour is the estimated 95\%
Gaussian population contour for the regional-cue and sine slopes; the dashed
contour has the same marginal standard deviations but zero correlation. Gray
points are the 150 shrunken conditional station modes, and the three coloured
points mark the displayed stations: circles, triangles and squares for ranks
4, 16 and 27. Their curves are solid, long-dashed and dot-dashed, respectively; $\rho$ in the panel title is the estimated
correlation. These standardized
curves are not observed episodes or uncertainty bands, and neither contour
is a confidence region.}
\label{fig:wind-expanded-mechanisms}
\end{figure}

\subsection{Predictive gains and the contribution of covariance}
Table~\ref{tab:wind-expanded-main} gives scores on identical test records.
Persistence fixes the mean to the latest observed grid direction, with its
error distribution fitted on development data. The circular forest uses the
implementation of \citet{Lang2020} (R packages \texttt{circtree}~0.1-0 and
\texttt{disttree}~0.2-0), with training out-of-bag calibration of concentration
and contamination. Its prespecified computational sample has
134,504 records and 100 trees, whereas the parametric models use all 1,166,700
estimation records. The unequal budgets limit any algorithmic superiority claim.
\begin{table}[tb]
\centering\small
\caption{Three-hour forecasts on 259,807 common test records at 150 stations, January--August 2026.}\label{tab:wind-expanded-main}
\begin{tabular}{@{}lrrr@{}}
\toprule
Model & Parameters & Mean log score $\uparrow$ & Mean error ($^\circ$) $\downarrow$ \\
\midrule
Persistence & 20 & -2.76443 & 33.893 \\
Circular forest & -- & -2.74282 & 33.146 \\
C: common slopes & 67 & -2.72726 & 32.779 \\
D: independent slopes & 70 & -2.71637 & 32.507 \\
F: correlated slopes & 73 & -2.71637 & 32.507 \\
\bottomrule
\end{tabular}
\par\smallskip\begin{minipage}{.98\textwidth}\footnotesize
Parameters count outer population parameters, not conditional station modes. C, D and F share circular station rotations and the same error family. The forest uses 134,504 training records; parametric models use 1,166,700. Forest complexity is not represented by a comparable parameter count.
\end{minipage}
\end{table}

Adding station slopes improves mean log score by 0.01089 nats per observation,
with a paired 95\% calendar-block interval $[0.00993,0.01180]$. Mean angular
error falls from 32.779 to 32.507 degrees, a gain of 0.272 degrees
($[0.247,0.295]$). All five regional log-score intervals exclude zero, and
133 of 150 station point differences are positive
(Figure~\ref{fig:wind-expanded-gains}). The gain is consistent across this
network but modest on the angular scale. Model D also improves log score over
persistence by 0.04806 ($[0.04474,0.05118]$) and over the configured forest
by 0.02645 ($[0.02287,0.02985]$), subject to the unequal training budgets.

\begin{figure}[!t]
\centering\includegraphics[width=\textwidth]{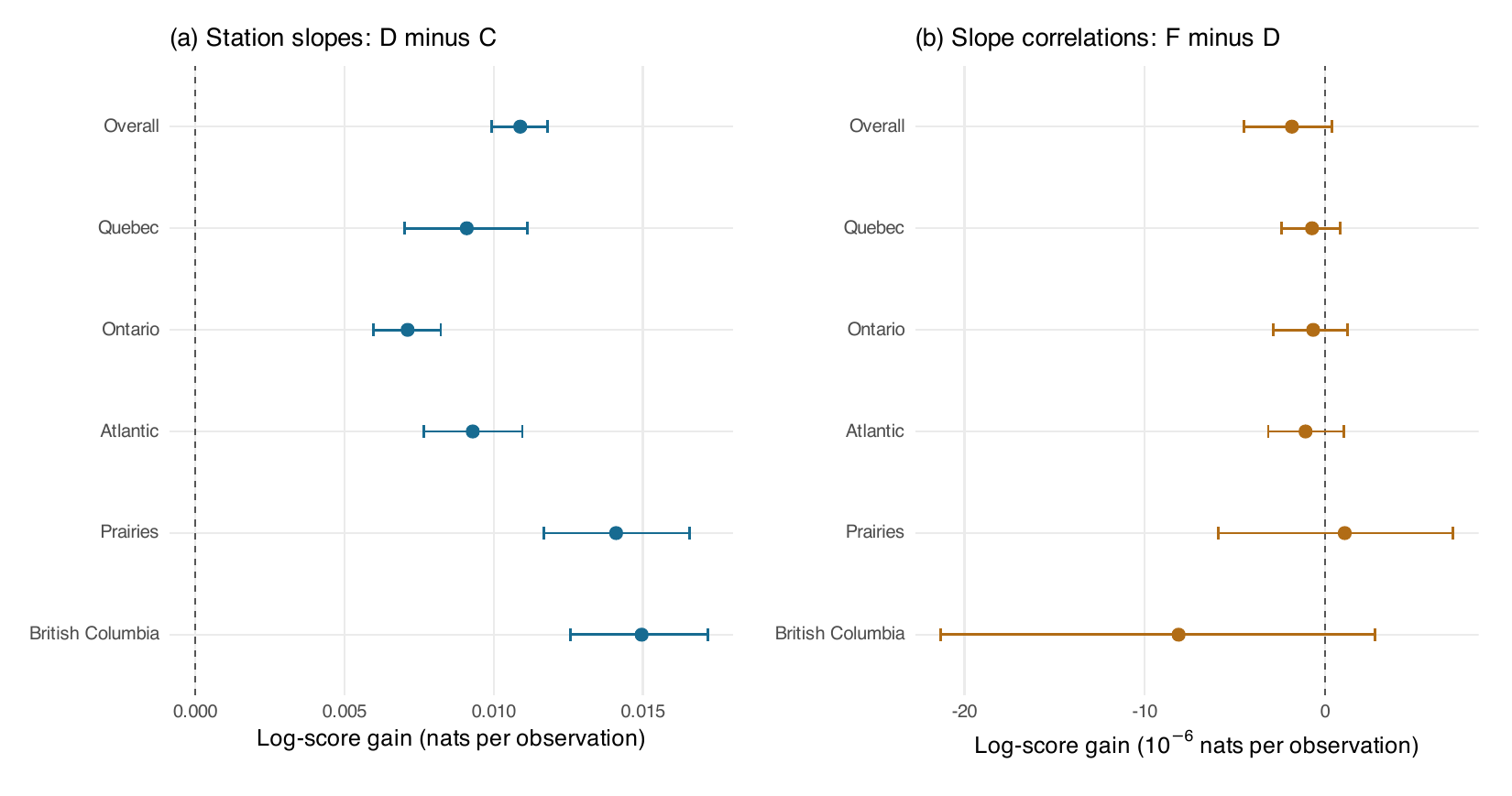}
\caption{New temporal evaluation at 150 Canadian stations. Left: independent
station slopes versus common slopes, with circular rotations in both models.
Right: the additional contribution of slope correlations. Points are paired
mean log-score gains; bars are 95\% percentile intervals from 5,000 resamples
of common 14-day calendar blocks. Intervals condition on the fitted models and
selected network and are not simultaneous. The two horizontal scales differ;
right-panel axis values are expressed in $10^{-6}$ nats per observation.}
\label{fig:wind-expanded-gains}
\end{figure}

Estimating the three correlations provides no useful additional predictive
gain: F minus D has mean log-score difference
$-1.84\times10^{-6}$, with interval
$[-4.48,0.38]\times10^{-6}$. The unrounded mean angular error increases by
only $0.000078$ degrees. Although the corresponding 14-day-block interval
excludes zero, its seven-day-block counterpart includes zero and the
magnitude is negligible. The two mean forecasts differ by only 0.0087 degrees
on average. The estimated correlations are 0.067, $-0.016$ and $-0.137$ for
regional--sine, regional--cosine and sine--cosine slope pairs. These are point
estimates, not evidence that population correlations are zero or an equivalence
test. Their model-based Laplace--Wald 95\% intervals, computed from the saved
Hessian after the evaluation, all contain zero and are about 0.33 wide; the
approximate likelihood-ratio statistic for F against D is 3.24 on three
parameters (Supplement Section~S6). The undercoverage in
Section~\ref{sec:simulation} concerns a different two-coefficient model;
that study does not calibrate these weather-model intervals. The reported
widths therefore describe model-based imprecision, without establishing
nominal coverage. D is the more parsimonious
choice for the particular known-station prediction task. The near-zero score
difference does not show that the covariance is irrelevant to uncertainty in
fixed effects, population variance components or station effects. We have not
compared interval coverage for C, D and F on these weather data.

Intervals use 5,000 paired resamples of 18 common 14-day calendar blocks,
retaining simultaneous episodes across stations within each block. They are
conditional on the fitted models and observed network and omit training and
station-selection uncertainty. Seven-day blocks and equal-station weighting
give the same substantive conclusion about D versus C and F versus D.
Regional intervals are pointwise, with no multiplicity adjustment. The
supplement reports all comparisons and station-level differences.

\subsection{Calibration and scope}
The added slopes do not remove predictive miscalibration
(Figure~\ref{fig:wind-expanded-calibration}). For D, direction sets targeting
at least 90\% probability receive mean mass 90.94\% but cover 88.99\% of test
responses, a discrepancy of $-1.94$ percentage points. At the nominal 50\%
level, assigned mass is 55.40\% and observed coverage 56.21\%. F is almost
identical. Coverage is compared with actual assigned mass because the discrete
categories generally overshoot the nominal target. These are descriptive
calibration discrepancies, not coefficient-interval coverage results.

\begin{figure}[!t]
\centering\includegraphics[width=.83\textwidth]{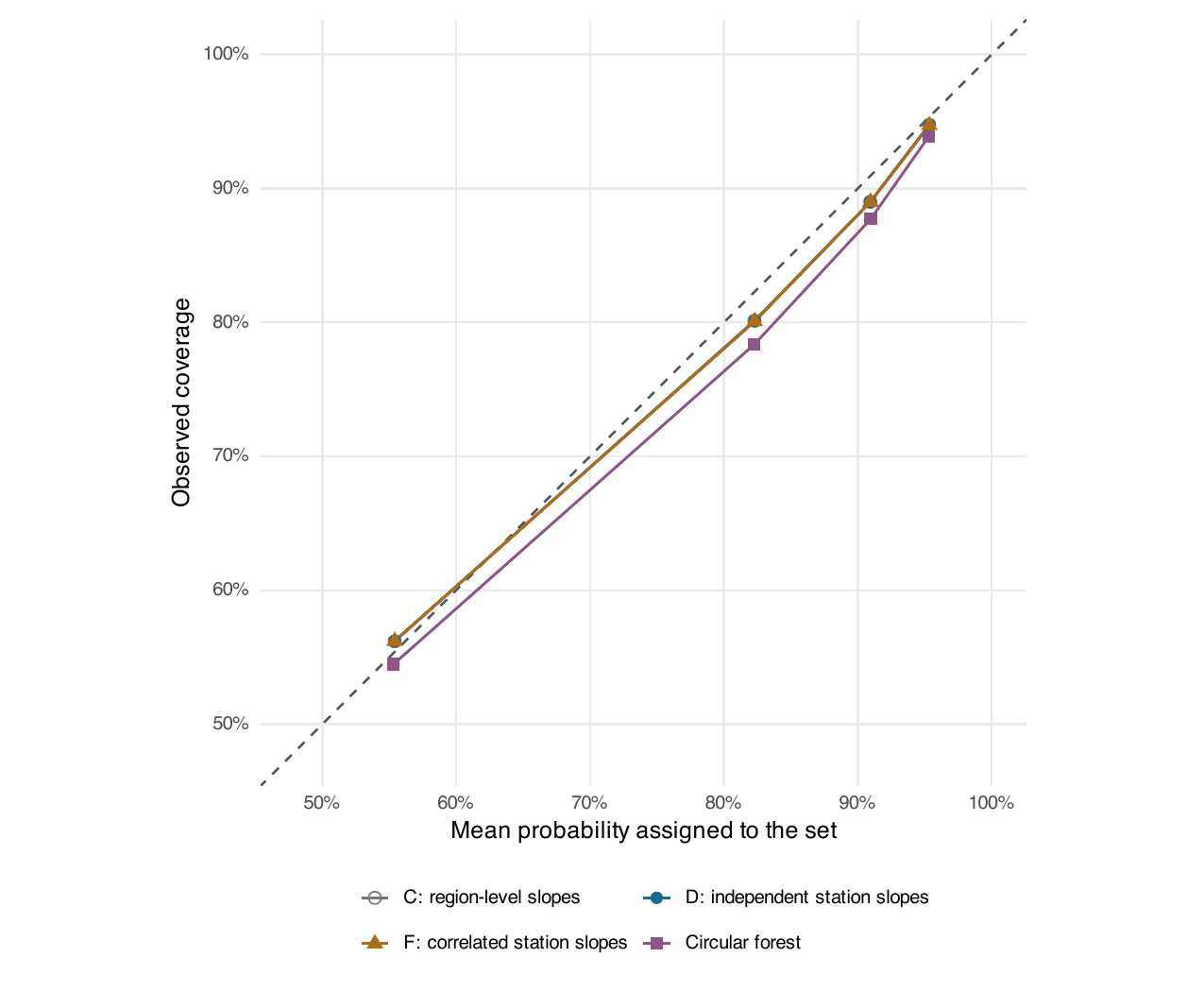}
\caption{Predictive calibration on the new 2026 test period. Points compare
observed coverage with mean assigned probability for discrete sets targeting
50\%, 80\%, 90\% and 95\%. The diagonal marks equality; categories are added
in decreasing probability until the target is reached. C, D and F are nearly
superimposed. The forest has a smaller training budget. These descriptive
curves have no uncertainty bands and do not establish conditional calibration.}
\label{fig:wind-expanded-calibration}
\end{figure}

The enlarged application supports a fitted description of station-level
heterogeneity after allowing for rotation and regional means. Estimating the
remaining slope correlations has little predictive effect in this task. The
empirical scores and prediction-set coverage assess forecasts, whereas
parameter standard errors and interval coverage require a separate inferential
assessment. The primary simulations address the latter question under their
specified data-generating model and show undercoverage for correlations. The
application does not identify physical terrain effects, validate a common
covariance within every region, or assess inference near resultant
cancellation. Its heteroskedastic error family and five-dimensional
station integral extend beyond the homogeneous-error theorem and the primary
simulation's two-coefficient design. Remaining calibration error, temporal and
spatial dependence, and use of fitted station modes qualify stronger claims
about operational forecasting or parameter uncertainty.

\section{Discussion}\label{sec:discussion}
Homogeneous angular regression separates the geometry of cue combination from
observation concentration. For repeated directional responses, this makes it
possible to distinguish common rotations from correlated changes in cue
responsiveness. The constructive theorem identifies these components under
specific designs, while coefficient-space integration provides smooth Gaussian
population scores despite irregular angular geometry. The identification
argument is nonparametric in the coefficient law; the fitted model and
algorithm are Gaussian.

The numerical results delimit what this construction currently supports.
Adapted integration improves the evaluated likelihood approximations, and
nearly all primary simulation fits and intervals are available. Correlation
interval calibration remains inadequate in several scenarios, especially with
few individuals. The saved likelihood and standard-error diagnostics cannot
determine whether this reflects finite-sample likelihood shape, information
estimation, shared numerical error or an undiscovered competing maximum.
A modified cutoff or alternative interval procedure would require a specified
rule and a separate evaluation. Weak directional designs and temporal
misspecification present further difficulties even when population identification
or numerical convergence holds.

The expanded wind application provides a new temporal comparison after
freezing the models. With circular rotation already included, independent
station slopes improve prediction over region-level slopes; the benefit is
modest but present in all five regions. Estimating slope correlations adds no
useful predictive gain, and their approximate intervals are wide. The
comparison supplies, in a new network, the comparator with station rotation
but no station slopes that the earlier Quebec analysis lacked; it does not determine the smallest adequate
random-slope subset. Forecast scores assess predictions at known stations.
They do not validate standard errors or interval coverage for the population
coefficients and covariance. The simulation study evaluates such inference
under its own two-coefficient model, not under the weather model. These forecasts
do not address transfer to unseen stations or inference near
cancellation. Unequal forest training budgets and remaining predictive
miscalibration further limit the application's conclusions.

The secondary analyses in Supplement Section~S8 reinforce the importance of
mean specification, initial-history conditioning and the response family:
restored mean terms change the sandhopper interpretation, the tested homogeneous
Maps submodel loses to a projected-normal comparator, and the movement
applications retain temporal-model limitations. Future work should prioritize
calibrated inference with few groups, predictive calibration and validation
across new stations and complete annual cycles.
Experiments with assigned directional cues would provide a direct setting for
the theorem's exogenous-design conditions.

\section*{Data and code availability}
The mathematical and numerical supplement is supplied as \path{supplement.pdf}.
The separate \path{reproducibility.zip} reconstruction bundle is retained in
the author-held research archive and is not included in this arXiv deposit. The archive contains saved assessment records,
frozen protocols, source code and the inputs used to rebuild the reported
simulation summaries, the complete fixed-parameter integration comparison and
weather-response illustrations without refitting. It also contains the expanded
Canadian weather results, frozen parametric fitted objects, individual test scores and
source programs needed to reconstruct the new tables and figures. Checks
added after the evaluation are separated in \path{new_analyses/}, with
their scripts and outputs. A separate portable launcher
loads the contaminated-model implementation; its example distinguishes saved
data inspection from optional re-estimation.
Its archived analyses retain secondary applications and adverse results.
The archive README identifies supplied inputs and distinguishes reconstruction
of reported summaries from full re-estimation using the larger source archives.
The weather observations originate from Environment and Climate Change Canada;
the expanded study is documented in Supplement Section~S6, with the earlier
Quebec application retained in Section~S7.

\section*{Funding}
This research received no funding.

\section*{Declaration of competing interest}
The author declares no competing interests.

\section*{Declaration of generative AI and AI-assisted technologies in the manuscript preparation process}
OpenAI Codex was used to assist with scientific review, editorial restructuring, language revision,
literature checking, preparation and checking of analysis code, and
reconstruction of tables and figures from saved results.
Anthropic Claude assisted with manuscript review, proposed textual and
mathematical clarifications, and prepared and ran the additional numerical
checks of saved results identified in Supplement Section~S6.
The implementation assistance included diagnostic scripts and the R code used
to construct tables and quantitative figures. The retained September 22--25
revision records identify Codex model versions \texttt{gpt-6-astra} and
\texttt{gpt-6-sol} (OpenAI); the retained Claude outputs do not identify a
model version. Figures are calculated from specified model functions or
archived numerical records. Executable code, frozen inputs, software versions
and the scope of the independent numerical checks are documented in the
separate reproduction archive and Supplement Section~S9. The author has reviewed and
edited the assisted content and takes full responsibility for the submitted
article.

\begingroup
\setlength{\bibsep}{0pt}
\bibliographystyle{elsarticle-harv}
\bibliography{references}
\endgroup
\end{document}